\documentclass[10pt,conference]{IEEEtran}
\IEEEoverridecommandlockouts
\usepackage{cite}
\usepackage{amsthm,amsmath,amssymb,amsfonts}
\usepackage{amssymb}
\usepackage{algorithmic}
\usepackage{graphicx}
\usepackage{textcomp}
\usepackage{xcolor}
\usepackage{amsmath}
\usepackage{subfigure}
\usepackage{epstopdf}
\usepackage{xcolor}
\IEEEoverridecommandlockouts
\def\BibTeX{{\rm B\kern-.05em{\sc i\kern-.025em b}\kern-.08em
    T\kern-.1667em\lower.7ex\hbox{E}\kern-.125emX}}
\begin{document}
\newtheorem{lemma}{Lemma}
\newtheorem{theorem}{Theorem}
\newtheorem{corollary}{Corollary}
\newtheorem{remark}{Remark}
\title{Mutual Information Metrics for Uplink MIMO-OFDM Integrated Sensing and Communication System\\

}

\author{\IEEEauthorblockN{Jinghui Piao$^1$, Zhiqing Wei$^{1,*}$, Xin Yuan$^2$, Xiaoyu Yang$^1$, Huici Wu$^1$, Zhiyong Feng$^1$}
	$^1$\text{Beijing University of Posts and Telecommunications, Beijing 100876, China}\\
	$^2$\text{Data61, CSIRO, Sydney, Australia}\\
	$^1$Email: \{piaojinghui, weizhiqing, xiaoyu.yang, dailywu, fengzy\}@bupt.edu.cn \\
	$^2$Email: \{xin.yuan\}@data61.csiro.au\\}
\maketitle

\begin{abstract}
As the uplink sensing has the advantage of easy implementation, it attracts great attention in integrated sensing and communication (ISAC) system. This paper presents an uplink ISAC system based on multi-input multi-output orthogonal frequency division multiplexing (MIMO-OFDM) technology. The mutual information (MI) is introduced as a unified metric to evaluate the performance of communication and sensing. In this paper, firstly, the upper and lower bounds of communication and sensing MI are derived in details based on the interaction between communication and sensing. And the ISAC waveform is optimized by maximizing the weighted sum of sensing and communication MI. The Monte Carlo simulation results show that, compared with other waveform optimization schemes, the proposed ISAC scheme has the best overall performance.
\end{abstract}


\section{Introduction}
Integrated Sensing and Communication (ISAC) is considered to be one of the most promising technologies in the future wireless network. By integrating communication and sensing into the same hardware and wireless resources \cite{feng2020joint}, ISAC systems can sense the environment while transmitting information on wireless channels, which can realize joint scheduling of communication resources and sensing resources, and improve spectrum, hardware and other resource utilization \cite{cheng2022hybrid}.

In the ISAC systems, the types of sensing consist of uplink and downlink sensing \cite{8827589}. Thus, with the differences of performance evaluation in communication and sensing, it is significant to investigate the performance metrics according to the different cases, including Cramer-Rao Bound (CRB) \cite{kumari2019adaptive}\cite{liu2021cramer}, Mean Square Error (MSE)\cite{kumari2019adaptive} as well as Mutual information (MI) \cite{zhang2019mutual,zhang2022joint,yuan2020spatio,ouyang2022performance}, etc.  MI has attracted great attention for waveform design as its similar physical meanings, formulations, and optimization methods between sensing and communication\cite{ouyang2022integrated}. Most of the researches on the MI for ISAC systems focus on downlink communication and sensing systems. Zhang \textit{et al.} \cite{zhang2019mutual} proposed an downlink OFDM waveform design in Gaussian mixture clutter by maximizing the radar MI under the constraint of channel capacity. In \cite{zhang2022joint}, the weighted sum of sensing MI and communication data rate was investigated for ISAC waveform and phase shift design in RIS-Assisted ISAC system. In \cite{yuan2020spatio}, waveform optimization by maximizing weighted sum of communication and sensing MI was proposed for downlink MIMO ISAC system.

 Few works elaborate on the information-theoretical performance metrics of uplink ISAC systems. Ouyang \textit{et al.} \cite{ouyang2022performance} analyzed the uplink ISAC system with mono-static radar, the communication part of which adopts a non-orthogonal multiple access (NOMA) protocol, and derived the exact expressions of MI as well as the approximations under high SNR. However, the researches on the MI-based waveform optimization of uplink ISAC system considering uplink sensing are rare.   

Hence, this paper derives the sensing and communication MI of the uplink MIMO-OFDM ISAC system, and ISAC waveform optimization based on the expressions of MI is proposed. It is considered that the communication and sensing have the same channel. As the sensing channel is unknown, the upper and lower bounds of communication MI are derived by regarding the sensing channel as the interference. Based on the demodulation signal, the upper and lower bounds of sensing MI are derived considering the demodulation error. Finally, the ISAC waveform is optimized based on the weighted sum of sensing and communication MI. The simulation results show the reasonability of the proposed ISAC waveform optimization scheme and the trade-off in performance between communication and sensing.

The rest of this paper is organized as follows. In Sec. \uppercase\expandafter{\romannumeral2}, the uplink MIMO-OFDM ISAC system model is described. In Sec. \uppercase\expandafter{\romannumeral3}, the upper and lower bounds of communication MI and sensing MI are derived. Sec. \uppercase\expandafter{\romannumeral4} proposes a waveform optimization scheme based on maximizing the weighted sum of derived MIs. In Sec. \uppercase\expandafter{\romannumeral5}, we provide the simulation results. 

\section{System Model}

\begin{figure}[!t]
	\centering
	\includegraphics[width=0.4\textwidth]{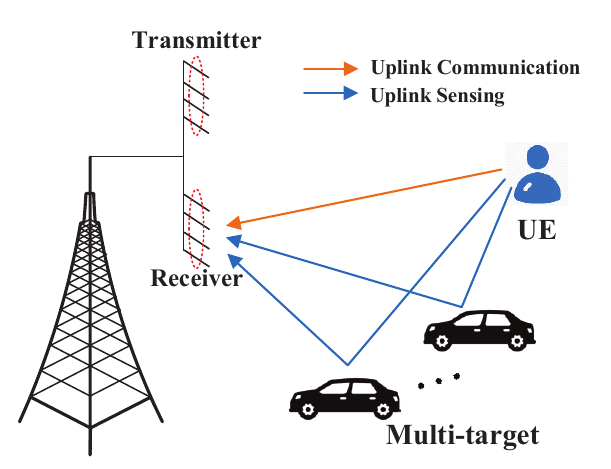}
	\caption{MIMO-OFDM ISAC system.}
	\label{fig}
\end{figure}
A MIMO-OFDM ISAC system is considered in this paper, which has a user device (UE) and multiple radar targets, as shown in Fig. 1. The base station (BS) consists of two spatially separated antenna arrays, i.e., ${{N}_{t}}$ transmit antenna elements and ${{N}_{r}}$ receive antenna elements. We consider uplink system where BS receives signals from UE for sensing and communication. The BS first regards the sensing as the interference to decode the communication signals. And the decoding signals are utilized to sense the environment\cite{zhang2021enabling}. UE is equipped with ${{N}_{t}}$ transmit antenna elements, and each transmit antenna transmits OFDM signals of ${{N}_{c}}$ subcarriers. 

Let $x_{n}^{\mu }\left( p\right) $ represent the data on the $\mu$-th transmit antenna element during the $n$-th symbol interval time on the $p$-th subcarrier. $\left\lbrace x_{n}^{\mu }(p), p=0,1,\cdots ,{{N}_{c}}-1\right\rbrace $ is processed and sent to the fading channel. And ${{h}_{\mu \nu }}(l)$ and ${{g}_{\mu \nu }}(l)$ denotes the coefficient of time domain communication channel ${{\mathbf{H}}_{l}}\in {{\mathbb{C}}^{{{N}_{t}}\times {{N}_{r}}}}$ and sensing channel ${{\mathbf{G}}_{l}}\in {{\mathbb{C}}^{{{N}_{t}}\times {{N}_{r}}}}$ between the $\mu$-th transmit antenna element
and the $\nu$-th receive antenna element on the $l$-th path, respectively. 

After the Fast Fourier Transform (FFT) processing, the received data from UE on the  $\nu$-th receive antenna of BS can be expressed as
\begin{equation}
y_{n}^{\nu }(p)=\sum\limits_{\mu =1}^{{{N}_{t}}}{{{H}_{\mu \nu }}}(p)x_{n}^{\mu }(p)+\sum\limits_{\mu =1}^{{{N}_{t}}}{{{G}_{\mu \nu }}(p)x_{n}^{\mu }(p)}+w_{n}^{\nu }(p)\label{eq1},	
\end{equation}
where ${{H}_{\mu \nu }}\left( p\right) =\sum\nolimits_{l=0}^{{{L}_{c}}}{{{h}_{\mu \nu }}(l){{e}^{-j(2\pi /{{N}_{c}})lp}}}$, ${{G}_{\mu \nu }}(p)=\sum\nolimits_{l=0}^{{{L}_{s}}}{{{g}_{\mu \nu }}(l){{e}^{-j(2\pi /{{N}_{c}})lp}}}$. $L_c$ and $L_s$ represents the communication paths and sensing paths, respectively. $w_{n}^{\nu }(p)$ represents the noise matrix, which follows independent and identically distributed (i.i.d.) zero-mean circularly symmetric complex Gaussian (CSCG) distribution with variance $\sigma _{n}^{2}$, i.e., $\mathcal{C}\mathcal{N}(0,\sigma _{n}^{2})$.

Hence, the received signal at the $p$-th subcarrier, denoted by $\mathbf{Y}(p)$, can be expressed as
\begin{equation}
\mathbf{Y}(p)=\mathbf{X}(p)\mathbf{H}(p)+\mathbf{X}(p)\mathbf{G}(p)+\mathbf{W}(p)\label{eq2},
\end{equation}
where
$\mathbf{H}(p)=\sum\nolimits_{l=0}^{L_c}{{{\mathbf{H}}_{l}}{{e}^{-j(2\pi /{{N}_{c}})lp}}}\in {{\mathbb{C}}^{{{N}_{t}}\times {{N}_{r}}}},$
$\mathbf{G}(p)=\sum\nolimits_{l=0}^{L_s}{{{\mathbf{G}}_{l}}{{e}^{-j(2\pi /{{N}_{c}})lp}}}\in {{\mathbb{C}}^{{{N}_{t}}\times {{N}_{r}}}},$\\
$\mathbf{X}(p)=\left[ \begin{matrix}
x_{0}^{1}(p) & \ldots  & x_{0}^{{{N}_{t}}}(p)  \\
\vdots  & {} & \vdots   \\
x_{{{N}_{x}}-1}^{1}(p) & \ldots  & x_{{{N}_{x}}-1}^{{{N}_{t}}}(p)  \\
\end{matrix} \right]\in {{\mathbb{C}}^{{{N}_{x}}\times {{N}_{t}}}},$\\
$\mathbf{Y}(p)=\left[ \begin{matrix}
y_{0}^{1}(p) & \ldots  & y_{0}^{{{N}_{r}}}(p)  \\
\vdots  & {} & \vdots   \\
y_{{{N}_{x}}-1}^{1}(p) & \ldots  & y_{{{N}_{x}}-1}^{{{N}_{r}}}(p)  \\
\end{matrix} \right]\in {{\mathbb{C}}^{{{N}_{x}}\times {{N}_{r}}}},$\\
$\mathbf{W}(p)=\left[ \begin{matrix}
w_{0}^{1}(p) & \ldots  & w_{0}^{{{N}_{r}}}(p)  \\
\vdots  & {} & \vdots   \\
w_{{{N}_{x}}-1}^{1}(p) & \ldots  & w_{{{N}_{x}}-1}^{{{N}_{r}}}(p)  \\
\end{matrix} \right]\in {{\mathbb{C}}^{{{N}_{x}}\times {{N}_{r}}}}.$ \\
where $\mathbf{H}(p)$ and $\mathbf{G}(p)$ are communication and sensing frequency domain channel coefficient matrix at the  $p$-th subcarrier, respectively. Thus, the received signal for ${{N}_{c}}$ subcarriers is
\begin{equation}
\mathbf{Y}=\mathbf{X}\mathbf{H}+\mathbf{X}\mathbf{G}+\mathbf{W}\label{eq3},
\end{equation}
where
$\mathbf{X}=\text{diag}\left\lbrace \mathbf{X}(0),\cdots ,\mathbf{X}({{N}_{c}}-1)\right\rbrace \in {{\mathbb{C}}^{{{N}_{c}}{{N}_{x}}\times {{N}_{c}}{{N}_{t}}}},$
$\mathbf{G}={{\left[ {{\mathbf{G}}^{T}}(0),\cdots ,{{\mathbf{G}}^{T}}({{N}_{c}}-1)\right] }^{T}}\in {{\mathbb{C}}^{{{N}_{c}}{{N}_{t}}\times {{N}_{r}}}},$\\
$\mathbf{H}={{\left[ {{\mathbf{H}}^{T}}(0),\cdots ,{{\mathbf{H}}^{T}}({{N}_{c}}-1)\right] }^{T}}\in {{\mathbb{C}}^{{{N}_{c}}{{N}_{t}}\times {{N}_{r}}}},$\\
$\mathbf{Y}={{\left[ {{\mathbf{Y}}^{T}}(0),\cdots ,{{\mathbf{Y}}^{T}}({{N}_{c}}-1)\right] }^{T}}\in {{\mathbb{C}}^{{{N}_{c}}{{N}_{x}}\times {{N}_{r}}}},$\\
and $\mathbf{W}={{\left[ {{\mathbf{W}}^{T}}(0),\cdots ,{{\mathbf{W}}^{T}}({{N}_{c}}-1)\right] }^{T}}\in {{\mathbb{C}}^{{{N}_{c}}{{N}_{x}}\times {{N}_{r}}}}.$

Note that for the uplink system, the channel are the same for communication and sensing\cite{zhang2021enabling}. Thus, it is assumed that ${{\mathbf{H}}_{l}}$ and ${{\mathbf{G}}_{l}}$ contains different paths of channel \cite{ni2022uplink}. It is also assumed that only the spatial correlation of transmit antenna is considered, and the channel matrix of different path is independent. Hence, ${{\mathbf{G}}_{l}}$ can be decomposed by
${{\mathbf{G}}_{l}}=\mathbf{R}_{s,l}^{\frac{1}{2}}{{\mathbf{G}}_{\omega ,l}},$
where $\mathbf{R}_{s,l}$ represent the correlation matrix, and the elements in ${{\mathbf{G}}_{\omega ,l}}\in {{\mathbb{C}}^{{{N}_{t}}\times {{N}_{r}}}}$ follows the i.i.d. zero mean CSCG distribution with variance $\sigma _{l}^{2}$, i.e., $\mathcal{C}\mathcal{N}(0,\sigma _{l}^{2})$. Since ${{\mathbf{G}}_{l}}$ contains paths needed to be sensed, which is unknown, we only assume that the $\mathbf{R}_{s,l}$ has been known to the BS.  

\addtolength{\topmargin}{0.05in}

\section{Mutual Information}
In this section, the communication and sensing MI expressions are derived based on the MIMO-OFDM system model given in the Sec. \uppercase\expandafter{\romannumeral1}. With the superposed received signal of sensing and communication, MI based on the effect of sensing on communication (the effect of communication on sensing) is derived in details and the upper and lower bounds of communication and sensing MI are given. 

\subsection{Communication MI}\label{AA}
As the sensing part is regarded as interference, (\ref{eq3}) can be transformed into
\begin{equation} 
\mathbf{Y}={\mathbf{X}}\mathbf{H}+{{\mathbf{W}}_{1}}\label{eq4},	
\end{equation}
where ${{\mathbf{W}}_{1}}={\mathbf{X}}{\mathbf{G}}+{\mathbf{W}}$ is regarded as the new noise interference. Its mean is $\mathbf{0}$, and its covariance matrix can be expressed as
\begin{equation}
\begin{aligned}
 E\left[ {{({{\mathbf{W}}_{1}})}^{H}}{{\mathbf{W}}_{1}}\right]&=E\left[ {{\mathbf{G}}^{H}}{{\mathbf{X}}^{H}}{\mathbf{X}}{\mathbf{G}}\right] +E\left[ {{\mathbf{W}}^{H}}{\mathbf{W}}\right]  \\ 
 &=E\left[ {{\mathbf{G}}^{H}}{{\mathbf{X}}^{H}}{\mathbf{X}}\mathbf{G}\right] +{{N}_{x}}{{N}_{c}}\sigma _{n}^{2}{{\mathbf{I}}_{{{N}_{r}}}}.\label{eq5} \\ 
\end{aligned} 
\end{equation}
\begin{lemma}
	The covariance matrix (\ref{eq5}) can be transformed into
	\begin{equation}
	\begin{split}
	E\left[ {{({{\mathbf{W}}_{1}})}^{\mathbf{H}}}{{\mathbf{W}}_{1}}\right] &={{N}_{x}}{{\mathbf{R}}_{{{W}_{1}}}}\\
	&={\rm{tr}}(E\left[ {{\mathbf{X}}^{H}}{\mathbf{X}}\right] {{\Sigma }_{\mathbf{G}}}){{\mathbf{I}}_{{{N}_{r}}}}+{{N}_{x}}\sigma _{n}^{2}{{\mathbf{I}}_{{{N}_{r}}}}\label{eq6},
	\end{split}
	\end{equation}	
	where ${{\mathbf{R}}_{{{W}_{1}}}}=\left( {\rm{tr}}(E[{{\mathbf{X}}^{H}}{\mathbf{X}}]{{\Sigma }_{\mathbf{G}}})/{{N}_{x}}+\sigma _{n}^{2} \right){{\mathbf{I}}_{{{N}_{r}}}}$.
\end{lemma}
\begin{proof}
	The proof is provided in Appendix A.
\end{proof}

Communication MI is defined as the MI between the transmitted signal and the received signal under the knowledge of channel state information (CSI), which can be expressed as
\begin{equation}	I(\mathbf{X};\mathbf{Y}|\mathbf{H})=h(\mathbf{Y}|\mathbf{H})-h({{\mathbf{W}}_{1}})\label{eq7}.
\end{equation}

Affected by sensing, it is impossible to determine the distribution of ${{\mathbf{W}}_{1}}$. Thus, the upper and lower bounds of communication MI are given as follows.

\begin{theorem}
	The lower and upper bounds of communication MI are respectively 
	\begin{equation}
	I_{\rm{low}}^{\rm{com}}={{N}_{x}}{{\log }_{2}}\det \left( {{\mathbf{I}}_{{{N}_{r}}{{N}_{c}}}}+{{\mathbf{H}}^{H}}{{\Sigma }_{\mathbf{X}}}\mathbf{H}({{\mathbf{I}}_{{{N}_{c}}}}\otimes \mathbf{R}_{{{W}_{1}}}^{-1}) \right)\label{eq8},
	\end{equation}
	\begin{equation}
	\begin{aligned}
	I_{\rm{high}}^{\rm{com}}=&I_{\rm{low}}^{\rm{com}}+{{N}_{x}}{{\log }_{2}}\det ({{\mathbf{I}}_{{{N}_{c}}}}\otimes {{\mathbf{R}}_{{{W}_{1}}}})\\
	&-{{N}_{r}}{{\log }_{2}}\left[ \det \left( \mathbf{X}{{\Sigma }_{\mathbf{G}}}{{\mathbf{X}}^{H}}+\sigma _{n}^{2}{{\mathbf{I}}_{{{N}_{x}}{{N}_{c}}}}\right)  \right]\label{eq9}.
	\end{aligned}
	\end{equation}	 
\end{theorem}
\begin{proof}
	The proof for (\ref{eq8}) and (\ref{eq9}) is provided in Appendix B and C, respectively.
\end{proof}
\subsection{Sensing MI}
After decoding the signal, communication part will be subtracted from $\mathbf{Y}$, while the rest part will be used for sensing. However, the demodulation error $\mathbf{E}$ will also be introduced, namely
\begin{equation}
\mathbf{X}=\hat{\mathbf{X}}+\mathbf{E}\label{eq10},	
\end{equation}
where $\mathbf{E}=\text{diag}\{\mathbf{E}(0),\cdots ,\mathbf{E}({{N}_{c}}-1)\}\in {{\mathbb{C}}^{{{N}_{c}}{{N}_{x}}\times {{N}_{c}}{{N}_{t}}}}$, $\mathbf{E}(p)={{\left[ {{e}_{ij}} \right]}_{{{N}_{x}}\times {{N}_{t}}}}$. $\hat{\mathbf{X}}$ and $\mathbf{E}$ are independent of each other \cite{gomadam2007optimal}. (\ref{eq10}) shows that the effect of communication is introduced into sensing. Hence, (\ref{eq3}) is transformed into
\begin{equation}
{{\mathbf{Y}}_{\rm{rad}}}=\hat{\mathbf{X}}\mathbf{G}+\mathbf{E}\mathbf{G}+\mathbf{W}=\hat{\mathbf{X}}\mathbf{G}+{{\mathbf{W}}_{2}}\label{eq11},
\end{equation}
where the mean of ${{\mathbf{W}}_{2}}$ is $\mathbf{0}$, and its covariance matrix can be expressed as
\begin{equation}
\begin{aligned}
 E\left[ {{\mathbf{W}}_{2}}{{({{\mathbf{W}}_{2}})}^{H}}\right] &=E\left[ \mathbf{E}\mathbf{G}{{\mathbf{G}}^{H}}{{\mathbf{E}}^{H}}\right] +E\left[ \mathbf{W}{{\mathbf{W}}^{H}}\right]  \\ 
& =E\left[ \mathbf{E}\mathbf{G}{{\mathbf{G}}^{H}}{{\mathbf{E}}^{H}}\right] +{{N}_{r}}\sigma _{n}^{2}{{\mathbf{I}}_{{{N}_{x}}{{N}_{c}}}}\label{eq12}. \\ 
\end{aligned}
\end{equation}	
\begin{lemma}
	The covariance matrix (\ref{eq12}) can be transformed into
	\begin{equation}
	\begin{aligned}
	&E\left[ {{\mathbf{W}}_{2}}{{({{\mathbf{W}}_{2}})}^{H}}\right] 
	={{N}_{r}}{{\mathbf{R}}_{{{W}_{2}}}}\\
	&={\rm{tr}}\left( E\left[ \sum\nolimits_{l=0}^{L}{{{\mathbf{G}}_{l}}\mathbf{G}_{l}^{H}} \right]{{\Sigma }_{{\bar{\mathbf{E}}}}} \right){{\mathbf{I}}_{{{N}_{x}}{{N}_{c}}}}+{{N}_{r}}\sigma _{n}^{2}{{\mathbf{I}}_{{{N}_{x}}{{N}_{c}}}}\label{eq13},
	\end{aligned}
	\end{equation}
	where ${{\mathbf{R}}_{{{W}_{2}}}}=\left( {\rm{tr}}\left( E\left[ \sum\nolimits_{l=0}^{L}{{{\mathbf{G}}_{l}}\mathbf{G}_{l}^{H}} \right]{{\Sigma }_{{\bar{\mathbf{E}}}}} \right)/{{N}_{r}}+\sigma _{n}^{2} \right){{\mathbf{I}}_{{{N}_{x}}{{N}_{c}}}}$ and ${{\Sigma }_{{\bar{\mathbf{E}}}}}$ represents the covariance matrix of demodulation error at a single subcarrier.
	
\end{lemma}
\begin{proof}
	The proof is provided in Appendix D.
\end{proof}

The sensing MI is defined as the MI between the radar reflected signal and the sensing channel under the knowledge of the signal, thus, the expression of sensing MI is \cite{bell1993information}
	
\begin{equation}
I(\mathbf{G};{{\mathbf{Y}}_{\text{rad}}}|\mathbf{\hat{X}})=h({{\mathbf{Y}}_{\text{rad}}}|\mathbf{\hat{X}})-h({{\mathbf{W}}_{2}})\label{eq14}.
\end{equation}

Affected by the demodulation error, it is impossible to determine the distribution of ${{\mathbf{W}}_{2}}$. The upper and lower bounds of sensing MI are given as follows.

\begin{theorem}
	The lower and upper bounds of sensing MI are respectively
	\begin{equation}
	I_{\rm{low}}^{\rm{rad}}={{N}_{r}}{{\log }_{2}}\det \left( {{\mathbf{I}}_{{{N}_{x}}{{N}_{c}}}}+\hat{\mathbf{X}}{{\Sigma }_{\mathbf{G}}}{{{\hat{\mathbf{X}}}}^{H}}\mathbf{R}_{{{W}_{2}}}^{-1} \right)\label{eq15},
	\end{equation}
	\begin{equation}
	\begin{aligned}
	I_{\rm{high}}^{\rm{com}}=&I_{\rm{low}}^{\rm{rad}}+{{N}_{r}}{{\log }_{2}}\det ({{\mathbf{R}}_{{{W}_{2}}}})\\
	 &-{{N}_{x}}{{\log }_{2}}\left[ \det \left( {\bar{\mathbf{G}}^H}{{\Sigma }_{\mathbf{E}}}\bar{\mathbf{G}}+{{\mathbf{I}}_{{{N}_{r}}{{N}_{c}}}}\right)  \right]\label{eq16},
	\end{aligned}
	\end{equation}
	where ${{\Sigma}_{\mathbf{E}}}=E\left[ {{\mathbf{E}}^{H}}\mathbf{E}\right] /{{N}_{x}}$ and $\bar{\mathbf{G}}={\rm{diag}}\left( \mathbf{G}(p) \right)$.
\end{theorem}
\begin{proof}
	The proof of (\ref{eq15}) and (\ref{eq16}) are similar to the derivation in Appendix B and C, respectively.
\end{proof}

\section{MIMO-OFDM ISAC Waveform Optimization}

In this section, MIMO-OFDM ISAC waveform is optimized by maximizing the weighted sum of communication and sensing MIs, which can be formulated as
\begin{equation}
\begin{aligned}
& \underset{\mathbf{X}}{\mathop{\max }}\,\quad{{F}_{\omega }}=\frac{{{\omega }_{r}}}{{{F}_{r}}}I(\mathbf{G};{{\mathbf{Y}}_{\rm{rad}}}|\hat{\mathbf{X}})+\frac{1-{{\omega }_{r}}}{{{F}_{c}}}I(\mathbf{X};\mathbf{Y}|\mathbf{H}), \\ 
& \text{  s}\text{.t.}\quad{\rm{tr}}\left( \mathbf{X}{{\mathbf{X}}^{H}} \right)\le E,\label{eq17} \\ 
\end{aligned}
\end{equation}
where ${{\omega }_{r}}$ and ${{\omega }_{c}}$ are the weighting factors of sensing and communication, respectively. ${{F}_{r}}$ and ${{F}_{c}}$ are the maximum sensing MI and maximum communication MI, respectively.

In order to optimize (\ref{eq17}), we use singular value decomposition (SVD) to decompose the correlation matrix, i.e., ${{\Sigma }_{\mathbf{G}}}={{\mathbf{U}}_{\mathbf{G}}}{{\Lambda }_{\mathbf{G}}}\mathbf{U}_{\mathbf{G}}^{H}$ and $\mathbf{H}{{\mathbf{H}}^{H}}={{\mathbf{U}}_{\mathbf{H}}}{{\Lambda }_{\mathbf{H}}}\mathbf{U}_{\mathbf{H}}^{H}$, where ${{\mathbf{U}}_{\mathbf{G}}}$ and ${{\mathbf{U}}_{\mathbf{H}}}$ are unitary matrices, while ${{\Lambda }_{\mathbf{G}}}$ and ${{\Lambda }_{\mathbf{H}}}$ are diagonal matrices, satisfying
$
{{\Lambda }_{\mathbf{G}}}=\text{diag}\{{{\mu }_{11}},\cdots ,{{\mu }_{{{N}_{t}}{{N}_{c}}}}\}\label{eq18}
$
and $
{{\Lambda }_{\mathbf{H}}}=\text{diag}\{{{\lambda }_{11}},\cdots ,{{\lambda }_{{{N}_{t}}{{N}_{c}}}}\}\label{eq19}.
$

Define $\mathbf{\Xi }=[{{\xi }_{ij}}]$, satisfying $\mathbf{\Xi}=\mathbf{U}_{\mathbf{G}}^{H}{{\mathbf{X}}^{H}}\mathbf{X}{{\mathbf{U}}_{\mathbf{G}}}=\mathbf{U}_{\mathbf{H}}^{H}{{\Sigma }_{\mathbf{X}}}{{\mathbf{U}}_{\mathbf{H}}}$. Based on the Hadamard's inequality, we can obtain $\det (\mathbf{\Xi })\le \prod\nolimits_{i=1}^{{{N}_{t}}{{N}_{c}}}{{{\xi }_{ii}}}$. Hence, according to the Cauchy-Schwarz inequality,  can be transformed into
\begin{equation} 
\begin{aligned}
&{{\mathbf{R}}_{{{W}_{1}}}}=\left( {\rm{tr}}(E[{{\mathbf{X}}^{H}}\mathbf{X}]{{\Sigma }_{\mathbf{G}}})/{{N}_{x}}+\sigma _{n}^{2} \right){{\mathbf{I}}_{{{N}_{r}}}}\\
&\le \left( \left( E\sum\nolimits_{i=1}^{{{N}_{t}}{{N}_{c}}}{{{\mu }_{ii}}} \right)/{{N}_{x}}+\sigma _{n}^{2} \right){{\mathbf{I}}_{{{N}_{r}}}}=\sigma _{{{n}'}}^{2}{{\mathbf{I}}_{{{N}_{r}}}}\label{eq20}.
\end{aligned}
\end{equation}

Since the demodulation error rate of communication system is low enough \cite{zheng2017super}, it can be assumed that $\mathbf{E}=\mathbf{0}$, i.e., ${{\mathbf{R}}_{{{W}_{2}}}}=\sigma _{n}^{2}{{\mathbf{I}}_{{{N}_{x}}{{N}_{c}}}}$. Hence, the lower bounds of communication and sensing MI can be simplified respectively as
\begin{equation}
\begin{aligned}
I(\mathbf{X};\mathbf{Y}|\mathbf{H})\le \sum\limits_{i=1}^{{{N}_{t}}{{N}_{c}}}{\left\{ {{N}_{x}}{{\log }_{2}}({{\lambda }_{ii}}{{\xi }_{ii}}/\sigma _{{{n}'}}^{2}+1) \right\}}\label{eq21},
\end{aligned}
\end{equation}
\begin{equation}
\begin{aligned}
I(\mathbf{G};{{\mathbf{Y}}_{\rm{rad}}}|\mathbf{X})\le \sum\limits_{i=1}^{{{N}_{t}}{{N}_{c}}}{\left\{ {{N}_{r}}{{\log }_{2}}({{\mu }_{ii}}{{\xi }_{ii}}/\sigma _{n}^{2}+1) \right\}}\label{eq22}.
\end{aligned}
\end{equation}

(\ref{eq21}) and (\ref{eq22}) are obtained based on the Sylvester's determinant equation, i.e., $\det (\mathbf{AB}+\sigma _{n}^{2}{{\mathbf{I}}_{N}})={{(\sigma _{n}^{2})}^{n-m}}\det (\mathbf{BA}+\sigma _{n}^{2}{{\mathbf{I}}_{M}})$. Substituting (\ref{eq21}) and (\ref{eq22}) into (\ref{eq17}), it can be transformed into 
\begin{equation}
\begin{aligned}
&\underset{\Xi }{\mathop{\max }}\,&{{F}_{\omega }}(\mathbf{\Xi})&=\sum\limits_{i=1}^{{{N}_{t}}{{N}_{c}}}\left\{ \frac{1-{{\omega }_{r}}}{{{F}_{c}}}{{N}_{x}}{{\log }_{2}}({{\mu }_{ii}}{{\xi }_{ii}}/\sigma _{{n}'}^{2}+1)\right.\\
& & &\left.+\frac{{{\omega }_{r}}}{{{F}_{r}}}{{N}_{r}}{{\log }_{2}}({{\lambda }_{ii}}{{\xi }_{ii}}/\sigma _{n}^{2}+1) \right\}, \\ 
& \text{ s}\text{.t}\text{.}&{\rm{tr}}(\mathbf{\Xi })&\le E,\text{ }{{\xi }_{ii}}\ge 0,\text{ }1\le i\le {{N}_{c}}{{N}_{t}}\label{eq23}.  
\end{aligned}
\end{equation}

Utilizing Karush-Kuhn-Tucker (KKT) condition for optimization, the Lagrangian function is
\begin{equation}
\begin{aligned}
 L(\mathbf{\Xi})=&-\sum\limits_{i=1}^{{{N}_{t}}{{N}_{c}}}\left\{ \frac{1-{{\omega }_{r}}}{{{F}_{c}}}{{N}_{x}}{{\log }_{2}}({{\lambda }_{ii}}{{\xi }_{ii}}/\sigma _{{{n}'}}^{2}+1) \right. \\
	&\left.+\frac{{{\omega }_{r}}}{{{F}_{r}}}{{N}_{r}}{{\log }_{2}}({{\mu }_{ii}}{{\xi }_{ii}}/\sigma _{n}^{2}+1) \right\} \\ 
& +\alpha (\sum\limits_{i=1}^{{{N}_{t}}{{N}_{c}}}{{{\xi }_{ii}}}-E)++{{\gamma }_{i}}(-{{\xi }_{ii}})\label{eq24}.  
\end{aligned}
\end{equation}
where $\alpha$ is the Lagrangian multiplier.

According to (\ref{eq24}), setting the partial derivative ${{\nabla }_{{{\xi }_{ii}}}}L(\mathbf{\Xi })$ of ${{\xi }_{ii}}$ to 0, and supplementing other conditions.
\begin{equation}
\frac{(1-{{\omega }_{r}}){{N}_{x}}}{{{F}_{c}}\ln 2}\frac{{{\lambda }_{ii}}}{\sigma _{{{n}'}}^{2}+{{\lambda }_{ii}}{{\xi }_{ii}}}+\frac{{{\omega }_{r}}{{N}_{r}}}{{{F}_{r}}\ln 2}\frac{{{\mu }_{ii}}}{\sigma _{n}^{2}+{{\mu }_{ii}}{{\xi }_{ii}}}=\alpha -{{\gamma }_{i}}\label{eq25},
\end{equation}
\begin{equation}
\alpha \left( \sum\limits_{i=1}^{{{N}_{t}}{{N}_{c}}}{{{\xi }_{ii}}-E} \right)=0\label{}\label{eq26},	
\end{equation}
\begin{equation}
{{\gamma }_{i}}{{\xi }_{ii}}=0\label{eq27},
\end{equation}
\begin{equation}
\alpha \ge 0,{{\gamma }_{i}}\ge 0,1\le i\le {{N}_{t}}{{N}_{c}}\label{eq28}.
\end{equation}

Let ${{v}_{i}}=\frac{{{\lambda }_{ii}}}{\sigma _{{{n}'}}^{2}}$, ${{\varphi }_{i}}=\frac{{{\mu }_{ii}}}{\sigma _{n}^{2}}$, $\varepsilon =\frac{{{\omega }_{r}}{{N}_{r}}}{{{F}_{r}}\ln 2}$ and $\eta =\frac{(1-{{\omega }_{r}}){{N}_{x}}}{{{F}_{c}}\ln 2}$. Obviously, if ${{\xi }_{ii}}\ne 0$, ${{\gamma }_{i}}$ is 0. Thus, the optimal solution can be obtained as
\begin{equation}\label{eq29}
{{\xi }_{ii}}\!\!=\!\!\left\{ \begin{aligned}
& \frac{1}{2}\!\left[\!\frac{1}{\zeta }(\varepsilon +\eta )\!-\!(\frac{1}{{{\nu }_{i}}}\!+\!\frac{1}{{{\varphi }_{i}}})+\right.\\
&\left.\;\;\sqrt{{{[(\frac{1}{{{\nu }_{i}}}\!-\!\frac{1}{{{\varphi }_{i}}})\!+\!\frac{1}{\alpha }(\eta \!-\!\varepsilon )]}^{2}}\!+\!4\frac{\varepsilon \eta }{{{\alpha }^{2}}}} \right]^{+},\!\!\!\!
&{{\nu }_{i}}\ne 0, {{\varphi }_{i}}\ne 0 \\ 
&0,&{{\nu }_{i}}=0, {{\varphi }_{i}}=0 \\ 
\end{aligned} \right.,
\end{equation}
\begin{equation}\label{eq30}
\sum\limits_{i=1}^{{{N}_{t}}{{N}_{c}}}{{{\xi }_{ii}}}=E.
\end{equation}
where ${{[x]}^{+}}=\max \{x,0\}$.

The optimal $\alpha $ can be found through binary search in $0<1/\alpha <1/\min \{\varepsilon /(1/{{\nu }_{i}}+E)+\eta /(1/{{\varphi }_{i}}+E)\}$.

\begin{figure}[!t]
	\centering
	\includegraphics[width=0.5\textwidth]{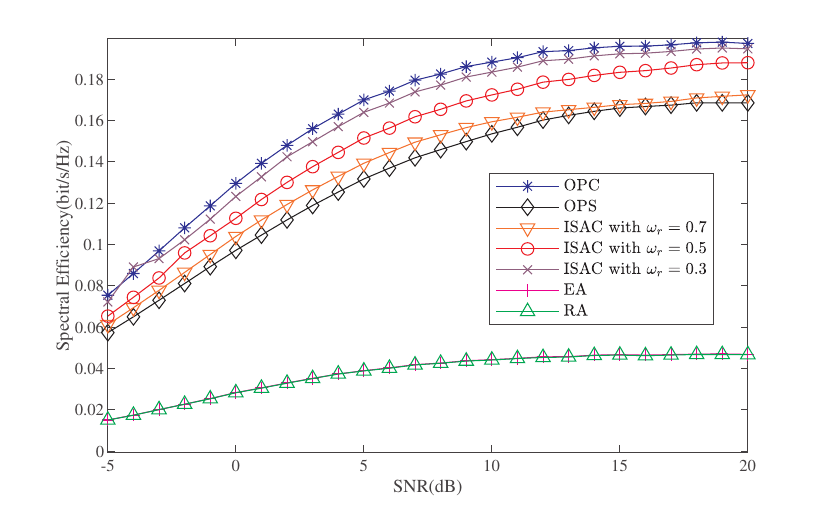}
	\caption{Spectral efficiency vs. SNR for different schemes.}
	\label{fig1}
\end{figure}
\begin{figure}[!t]
	\centering
	\includegraphics[width=0.5\textwidth]{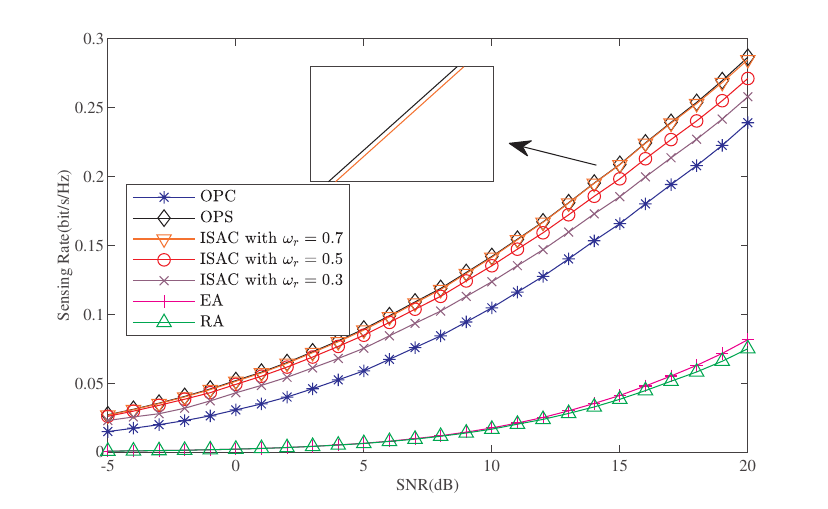}
	\caption{Sensing Rate vs. SNR for different schemes.}
	\label{fig2}
\end{figure}

\section{Simulation Results}
In this section, the above optimization results are simulated and analyzed, compared with other schemes, including optimal communication scheme (OPC), optimal sensing scheme (OPS), equal allocation (EA) and random allocation (RA). It is considered that the antenna arrays of UE and BS have 4 transmit antennas and 4 receive antennas, i.e., ${{N}_{t}}=4$ and ${{N}_{r}}=4$, respectively. The transmit antennas transmit OFDM signals with 32 IFFT points, i.e., ${{N}_{c}}=32$. The channels in this paper are assumed to be frequency-selective MIMO channels, and remain unchanged during the symbol duration with ${{N}_{x}}=10$. The channel is generated based on the Kronecker model with the maximum correlation coefficient of 0.5. And it is assumed that communication and sensing channel have 4 paths. The simulation results are obtained under the average of 4000 Monte Carlo simulations. 

Fig. \ref{fig1} and Fig. \ref{fig2} describe the spectral efficiency and sensing rate of different optimization schemes versus the SNR, respectively. With the increase of the SNR, the spectral efficiency and sensing rate of several schemes gradually increases. In general, the communication performance of the OPC scheme is much better than other schemes, while the sensing performance of the OPS is the best. The performance of the ISAC scheme with different sensing weighting factors are between OPC and OPS schemes, and EA and RA schemes are the worst. Additionally, with the sensing weighting factor increasing, the communication performance of ISAC gets closer to the OPS scheme, and on the contrary, the sensing performance of ISAC scheme nearly coincides with that of OPS scheme. ISAC scheme reduces the great impact of sensing on communication performance or communication on sensing performance in the integrated system.

Fig. \ref{fig3} shows the weighted MI of different schemes versus the sensing weighting factor, with SNR=1 dB. According to Fig. \ref{fig3}, it can be found that the total performance of ISAC increases at first and then decreases with the increase of the weighting factor, while other schemes remain unchanged. When ${{\omega }_{r}}=0$ and ${{\omega }_{r}}=1$, the total performance of the ISAC scheme coincides with that of the OPC and OPS schemes, respectively. 

Fig. \ref{fig4} shows the trade-off curve of ISAC scheme between the OPC scheme and the OPS scheme under different SNR. The sensing weighting factor ranges from 0 to 1, following the direction of the arrow. As the weighting factor increases, the ISAC scheme changes between the OPC scheme and the OPS scheme. Hence, in order to meet the specific uplink requirements, it is crucial to choose an appropriate sensing weighting factor.

\begin{figure}[!t]
	\centering
	\includegraphics[width=0.5\textwidth]{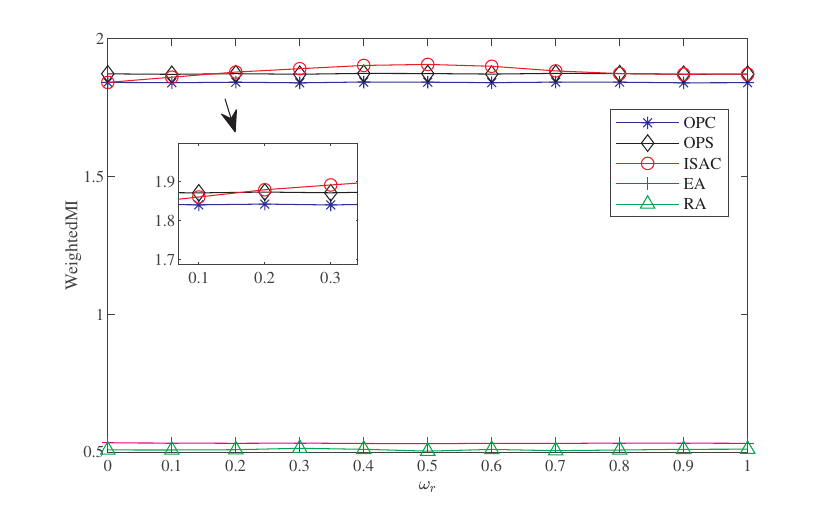}
	\caption{Weighted MI vs. weighting coefficient for different waveform optimization schemes}
	\label{fig3}
\end{figure}
\begin{figure}[!t]
	\centering
	\includegraphics[width=0.5\textwidth]{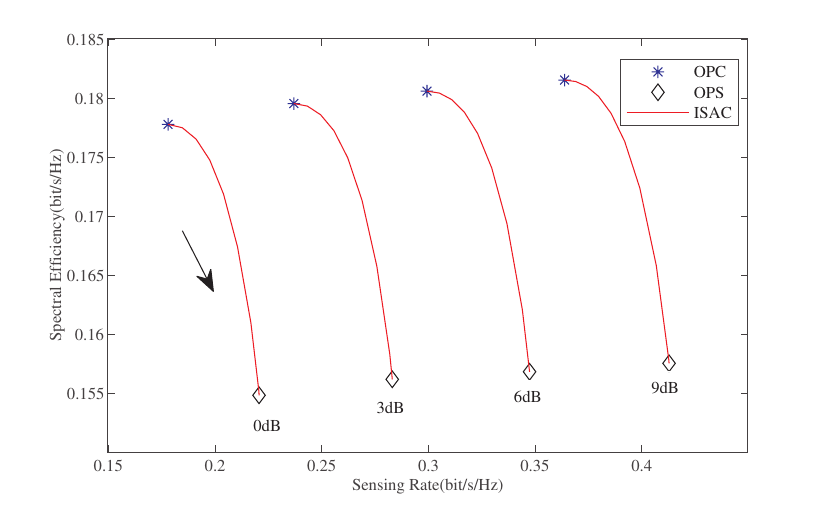}
	\caption{Trade-off curve.}
	\label{fig4}
\end{figure}

\section{Conclusion}
In this paper, the upper and lower bounds of sensing and communication MI for uplink MIMO-OFDM ISAC systems were studied considering the impact between communication and sensing and the waveform were optimized based on the weighted sum of communication and sensing MIs. The simulation results indicate the trade-off between communication and sensing in uplink ISAC system. Based on our work, the synchronization between BS and UE can be investigated in the future.

\begin{appendices} 
\section{Proof of Lemma1}\label{app:A}
\begin{proof}
Transforming $\mathbf{H}(p)$ and $\mathbf{G}(p)$ into time-domain form, we can obtain
\begin{equation}\label{eq32}
\mathbf{H}(p)=\mathbf{\Omega} (p)\cdot \mathbf{h},
\end{equation}
\begin{equation}\label{eq33}
\mathbf{G}(p)=\mathbf{\Omega} (p)\cdot \mathbf{g},
\end{equation}
where 
$\mathbf{\Omega} (p)={{I}_{{{N}_{t}}}}\otimes \mathbf{\omega} (p),$\\
$\mathbf{\omega} (p)=\left[ 1,{{e}^{-j2\pi p/{{N}_{c}}}},\ldots ,{{e}^{-j2\pi pL/{{N}_{c}}}} \right],$ \\
$\mathbf{h}=\left[ {{h}^{1}},\ldots ,{{h}^{{{N}_{r}}}} \right]\in {{\mathbb{C}}^{{{N}_{t}}(L+1)\times {{N}_{r}}}},$\\
${{h}^{\nu }}={{\left[ {{({{h}_{1\nu }})}^{T}},\ldots ,{{({{h}_{{{N}_{t}}\nu }})}^{T}} \right]}^{T}},$\\
${{h}_{\mu \nu }}={{\left[ {{h}_{\mu \nu }}(0),{{h}_{\mu \nu }}(1),\cdots ,{{h}_{\mu \nu }}(L) \right]}^{T}}\in {{\mathbb{C}}^{(L+1)\times 1}},$ and $\mathbf{g}$ has the same form as $\mathbf{h}$.

Thus, $\mathbf{H}$ and $\mathbf{G}$ can be expressed as
\begin{equation}\label{eq34}
\mathbf{H}=\mathbf{\Omega} \cdot \mathbf{h},
\end{equation}
\begin{equation}\label{eq35}
\mathbf{G}=\mathbf{\Omega} \cdot \mathbf{g},
\end{equation}
where
$\mathbf{H}={{[{{\mathbf{H}}^{T}}(0),\cdots ,{{\mathbf{H}}^{T}}({{N}_{c}}-1)]}^{T}}\in {{\mathbb{C}}^{{{N}_{c}}{{N}_{t}}\times {{N}_{r}}}},$
$\mathbf{G}={{[{{\mathbf{G}}^{T}}(0),\cdots ,{{\mathbf{G}}^{T}}({{N}_{c}}-1)]}^{T}}\in {{\mathbb{C}}^{{{N}_{c}}{{N}_{t}}\times {{N}_{r}}}},$
$\mathbf{\Omega} ={{[{{\mathbf{\Omega} }^{T}}(0),\ldots ,{{\mathbf{\Omega} }^{T}}({{N}_{c}}-1)]}^{T}}\in {{\mathbb{C}}^{{{N}_{c}}{{N}_{t}}\times {{N}_{t}}(L+1)}}.$

Hence, (\ref{eq5}) can be transformed by
\begin{equation}\label{eq36}
\begin{aligned}
& E\left[ {{({{\mathbf{W}}_{1}})}^{H}}{{\mathbf{W}}_{1}}\right]  \\ 
& =E\left[ {{\mathbf{G}}^{H}}{{\mathbf{X}}^{H}}\mathbf{X}\mathbf{G}\right] +{{N}_{x}}{{N}_{c}}\sigma _{n}^{2}{{\mathbf{I}}_{{{N}_{r}}}} \\ 
& =E\left\lbrace {{[{{\mathbf{g}}^{1}}\!\!,\ldots,\! {{\mathbf{g}}^{{{N}_{r}}}}]}^{H}}{{\mathbf{\Omega} }^{H}}{{\mathbf{X}}^{H}}\mathbf{X}\mathbf{\Omega} [{{\mathbf{g}}^{1}}\!\!,\ldots,\! {{\mathbf{g}}^{{{N}_{r}}}}]\right\rbrace +{{N}_{x}}{{N}_{c}}\sigma _{n}^{2}{{\mathbf{I}}_{{{N}_{r}}}} \\ 
& ={\rm{diag}}\left\{ {\rm{tr}}(E[{{\mathbf{\Omega} }^{H}}{{\mathbf{X}}^{H}}\mathbf{X}\mathbf{\Omega} ] E[{{\mathbf{g}}^{\nu}}{{({{\mathbf{g}}^{\nu}})}^{H}}]) \right\}+{{N}_{x}}{{N}_{c}}\sigma _{n}^{2}{{\mathbf{I}}_{{{N}_{r}}}}, \\ 
\end{aligned}
\end{equation}
where $E [{{\mathbf{g}}^{\nu }}{{({{\mathbf{g}}^{\nu }})}^{H}}]$ at different receive antenna is equivalent, and satisfies ${{\Sigma }_{\mathbf{G}}}=E [\mathbf{\Omega} {{\mathbf{g}}^{1}}{{({{\mathbf{g}}^{1}})}^{H}}{{\mathbf{\Omega} }^{H}}]=\cdots =E [\mathbf{\Omega} {{\mathbf{g}}^{{{N}_{r}}}}{{({{\mathbf{g}}^{{{N}_{r}}}})}^{H}}{{\mathbf{\Omega} }^{H}}]$, which is substituted into (\ref{eq36}) to get (\ref{eq6}).
\end{proof}

\section{Proof of Theorem 2}\label{app:B}
\begin{proof}
Rewritten the expression of the communication MI as
\begin{equation}\label{eq37}
\begin{aligned}
& I(\mathbf{X};\mathbf{Y}|\mathbf{H})=h(\mathbf{X}|\mathbf{H})-h(\mathbf{X}|\mathbf{Y},\mathbf{H}) \\ 
& =h(\tilde{\mathbf{X}}|\mathbf{h})-h(\tilde{\mathbf{X}}|\mathbf{Y},\mathbf{h}), \\ 
\end{aligned}
\end{equation}
where $\tilde{\mathbf{X}}=\mathbf{X}\mathbf{\Omega} $.

Obviously, the first term in (\ref{eq37}) follows complex Gaussian distribution which can be expressed as
\begin{equation}\label{eq38} h(\tilde{\mathbf{X}}|\mathbf{h})={{N}_{x}}\sum\limits_{p=0}^{{{N}_{c}}-1}{E\left\{ {{\log }_{2}}\pi e\det \left( {{\mathbf{Q}}_{p}} \right) \right\}},
\end{equation}
 where ${{\mathbf{Q}}_{p}}={{\mathbf{\Omega} }^{H}}(p){{\mathbf{X}}^{H}}(p)\mathbf{X}(p)\mathbf{\Omega} (p)/{{N}_{x}}$. The second term is constrained by the complex Gaussian distribution with
\begin{equation} \label{eq39}
\begin{aligned}
h(\tilde{\mathbf{X}}|\mathbf{Y},\mathbf{h})\le &{{N}_{x}}\sum\limits_{p=0}^{{{N}_{c}}-1}E\Big\{  {{\log }_{2}}\pi e\det \Big( {{\mathbf{Q}}_{p}}\\
&-{{\mathbf{Q}}_{p}}\mathbf{h}{{\left( {{\mathbf{h}}^{H}}{{\mathbf{Q}}_{p}}\mathbf{h}+{{\mathbf{R}}_{{{W}_{1}}}} \right)}^{-1}}{{\mathbf{h}}^{H}}{{\mathbf{Q}}_{p}} \Big) \Big\} .
\end{aligned}
\end{equation}

Hence, the lower bound of communication MI is
\begin{equation}
\begin{aligned}
 &I(\mathbf{X};\mathbf{Y}|\mathbf{H})\ge {{N}_{x}}\sum\limits_{p=0}^{{{N}_{c}}-1}{E\left\{ {{\log }_{2}}\det \left( {\mathbf{I}}_{{N}_{r}}+{{\mathbf{h}}^{H}}{{\mathbf{Q}}_{p}}\mathbf{h}\mathbf{R}_{{{W}_{1}}}^{-1} \right) \right\}} \\ 
& ={{N}_{x}}E\left\{ {{\log }_{2}}\det \left( {{\mathbf{I}}_{{{N}_{r}}{{N}_{c}}}}+{{\mathbf{H}}^{H}}{{\Sigma }_{\mathbf{X}}}\mathbf{H}({{\mathbf{I}}_{{{N}_{c}}}}\otimes \mathbf{R}_{{{W}_{1}}}^{-1}) \right) \right\}, \\ 
\end{aligned}
\end{equation}
where ${{\Sigma }_{\mathbf{X}}}=E\left\{ {{\mathbf{X}}^{H}}\mathbf{X} \right\}/{{N}_{x}}$, and $\mathbf{H}$ is transformed into $\mathbf{H}=\text{diag}\left\{ \mathbf{H}\left( p \right) \right\}$.
\end{proof}

\section{Proof of Theorem 2}
\begin{proof}
Rewritten the expression of the communication MI in another form as
\begin{equation}\label{eq41}
I(\mathbf{X};\mathbf{Y}|\mathbf{H})=h(\mathbf{Y}|\mathbf{H})-h(\mathbf{Y}|\mathbf{X},\mathbf{H}).
\end{equation}

Assuming that $\mathbf{Y}$ follows complex Gaussian distribution, thus, we can obtain
\begin{equation}\label{eq42}
h(\mathbf{Y}|\mathbf{H})\le {{N}_{x}}\sum\limits_{p=0}^{{{N}_{c}}-1}{E\left\{ {{\log }_{2}}\pi e\det \left( {{\mathbf{R}}_{{{W}_{1}}}}+{{\mathbf{h}}^{H}}{{\mathbf{Q}}_{p}}\mathbf{h} \right) \right\}}.
\end{equation}

Since $\mathbf{G}$ follows complex Gaussian distribution, the second term in (\ref{eq41}) can be expressed as
\begin{equation}\label{eq43}
\begin{aligned}
h(\mathbf{Y}|\mathbf{X},\mathbf{H}) &= h(\mathbf{G}\mathbf{X}+\mathbf{W}|\mathbf{X}) \\
&= {{N}_{x}}{{N}_{r}}{{N}_{c}}{{\log }_{2}}\pi +{{N}_{x}}{{N}_{r}}{{N}_{c}}{{\log }_{2}}e\\
& +{{N}_{r}}{{\log }_{2}}\left[ \det (\mathbf{X}{{\Sigma }_{\mathbf{G}}}{{\mathbf{X}}^{H}}+\sigma _{n}^{2}{{\mathbf{I}}_{{{N}_{x}}{{N}_{c}}}}) \right].
\end{aligned}
\end{equation}
 Substituting (\ref{eq42}) and (\ref{eq43}) into (\ref{eq41}) to obtain (\ref{eq9}).
\end{proof}
\section{Proof of Lemma 2}
\begin{proof}
Assuming that $\mathbf{E}$ is independent of each subcarrier, thus, (\ref{eq12}) can be transformed into  
\begin{equation}
\begin{aligned}
& E\left[ {{({{\mathbf{W}}_{2}})}^{H}}{{\mathbf{W}}_{2}}\right] =E\left[ \mathbf{E}\mathbf{G}{{\mathbf{G}}^{H}}{{\mathbf{E}}^{H}}\right] +{{N}_{r}}\sigma _{n}^{2}{{\mathbf{I}}_{{{N}_{x}}{{N}_{c}}}} \\ 
& =E\left\lbrace {{\left[ {{\mathbf{e}}^{1}},\ldots ,{{\mathbf{e}}^{{{N}_{r}}}}\right] }^{H}}\mathbf{G}{{\mathbf{G}}^{H}}\left[ {{\mathbf{e}}^{1}},\ldots ,{{\mathbf{e}}^{{{N}_{r}}}}\right] \right\rbrace +{{N}_{r}}\sigma _{n}^{2}{{\mathbf{I}}_{{{N}_{x}}{{N}_{c}}}} \\ 
& = {\rm{diag}} \left\lbrace  {\mathbf{R}}'(k) \right\rbrace +{{N}_{r}}\sigma _{n}^{2}{{\mathbf{I}}_{{{N}_{x}}{{N}_{c}}}}, \\ 
\end{aligned}
\end{equation}
where ${\mathbf{R}}'(k)=tr(E[\mathbf{G}(p){{\mathbf{G}}^{H}}(p)]{{\Sigma }_{\mathbf{E}(p)}}),k=(p+1)h$, $p=0,1,\cdots ,{{N}_{c}}-1$, $h=1,\cdots ,{{N}_{t}}$, ${{\Sigma }_{\mathbf{E}(p)}}=E [{{\mathbf{E}}^{H}}(p)\mathbf{E}(p)]/{{N}_{x}}$. Expanding $E[\mathbf{G}(p){{\mathbf{G}}^{H}}(p)]$ as
\begin{equation}\label{eq45}
\begin{aligned}
&E\left[ \mathbf{G}(p){{\mathbf{G}}^{H}}(p) \right]\\
&=E\left[ \sum\nolimits_{l=0}^{L}{{{\mathbf{G}}_{l}}}{{e}^{-j(2\pi /{{N}_{c}})lp}}\sum\nolimits_{l=0}^{L}{\mathbf{G}_{l}^{H}}{{e}^{j(2\pi /{{N}_{c}})lp}} \right].
\end{aligned}
\end{equation}

Since the fading of each path is independent of each other and follows CSGN distribution with zero-mean, i.e., $E[{{g}_{\mu \nu }}(l){{g}_{\mu \nu }}({l}')]=0$, thus,
\begin{equation}\label{eq46}
E\left[ \mathbf{G}(p){{\mathbf{G}}^{H}}(p) \right]=E\left[ \sum\nolimits_{l=0}^{L}{{{\mathbf{G}}_{l}}\mathbf{G}_{l}^{H}} \right].
\end{equation}

Assuming that the demodulation error correlation matrices of different subcarriers are independent and follow identical distribution, i.e., ${{\Sigma }_{\mathbf{E}(p)}}={{\Sigma }_{{\bar{\mathbf{E}}}}}$, we get (\ref{eq13}).
\end{proof}
\end{appendices}

\section*{Acknowledgment}
This work was supported in part by the National Natural Science Foundation of China (NSFC) under Grant 92267202,
in part by the National Key Research and Development Program of China under Grant 2020YFA0711302, and in part by the National Natural Science Foundation of China (NSFC) under Grant 62271081, and U21B2014.

\bibliographystyle{IEEEtran}
\bibliography{wenxian}

\end{document}